\documentclass{article}
\usepackage{amsmath,amssymb,graphicx,epsfig,cite,rotating,setspace,amsthm,color}

\newcommand{\beq}{\begin{equation}}  
\newcommand{\eeq}{\end{equation}}  
\newcommand{\bea}{\begin{eqnarray}}  
\newcommand{\eea}{\end{eqnarray}}  
\newcommand\la{{\lambda}}  
  
\newcommand\al{{\alpha}}

\newtheorem{Thm}{Theorem} 
\newtheorem{Def}{Definition}  
\newtheorem{Lem}{Lemma}  

\newcommand\K{{\mathrm{K}}} 

\newcommand\R{{\mathbb{R}}}

\begin{document}
%%%%%%%%%%%%%%%%%%%%%%%%%%%%%%%
%%%%%%%%%%%%%%%%%%%%%%%%%%%%%%%
\title
{Stability of 
stationary 
solutions for nonintegrable peakon equations}
%%%%%%%%%%%%%%%%%%%%%%%%%%%%%%%
%%%%%%%%%%%%%%%%%%%%%%%%%%%%%%%
\author
{A.N.W. Hone\thanks{School of Mathematics, Statistics \& Actuarial Science, University of Kent, Canterbury, Kent, U.K.} 
$\,$and S. Lafortune\thanks{Department of Mathematics, College of Charleston, 66 George Street, Charleston, South Carolina, 29424, U.S. 
Email: lafortunes@cofc.edu, Phone: 843-953-5869, Fax: 843-953-1410.}}

\maketitle

\begin{abstract} 
The Camassa-Holm equation with linear dispersion was originally derived as an asymptotic 
equation in shallow water wave theory. Among its many interesting mathematical properties, 
which include complete integrability,  
perhaps the most striking is the fact that in the case where linear dispersion is absent 
it admits weak multi-soliton solutions - ``peakons'' - with a peaked shape corresponding 
to a discontinuous first derivative. There is 
a one-parameter family of generalized Camassa-Holm equations, most of which are not 
integrable, but which all admit peakon solutions. Numerical studies reported by Holm and 
Staley indicate changes in the stability of these and other 
solutions as the parameter varies through the family. 

In this article, we describe 
 analytical results on one of these bifurcation phenomena, 
showing that in a suitable parameter range there are stationary solutions   
- ``leftons'' - 
which are orbitally stable.   
\end{abstract} 

\section{Introduction} 
The family of partial differential equations 
\beq \label{bfamily} 
u_t - u_{xxt} +(b+1)uu_x=bu_x u_{xx} + u u_{xxx}, 
\eeq 
labelled by the parameter $b$, 
is distinguished by the fact that it includes two completely integrable 
equations, namely the Camassa-Holm equation (the case $b=2$ \cite{ch, ch2}), and the 
Degasperis-Procesi equation (the case $b=3$ \cite{dp, dhh}). Each of the two  
integrable equations arises as the compatibility condition 
for an associated pair of linear equations (a Lax pair), 
and the latter leads to other hallmarks of integrability, 
namely the inverse scattering transform, multi-soliton solutions 
\cite{Ma05,honejpa,ch2},  
an infinite number of conservation laws, and a bi-Hamiltonian 
structure. (The latter structure for the case $b=2$ was 
found in \cite{ff}.)  According to various tests for integrability, 
the cases $b=2,3$ are the only integrable equations within this family 
\cite{dp, wanghone, miknov, honep}.  

The Camassa-Holm equation was originally proposed 
as a model for shallow water waves \cite{ch, ch2}, and it is explained in 
\cite{dgh1, dgh} that the members of the  family of equations 
(\ref{bfamily}), apart from the case $b=-1$, are asymptotically equivalent by means of an 
appropriate Kodama transformation. The results of \cite{constantin} (see Proposition 2 therein, 
and also equation (3.8) in \cite{rossen}) show that, 
in a model of shallow water flowing over a flat bed, 
the solution $u$ of (\ref{bfamily}) corresponds 
to the horizontal component of velocity evaluated at the level line 
$\theta\in [0,1]$, where 
$ 
\theta = \sqrt{\frac{11b-10}{12b}}, 
$ 
which requires either $b\geq  10/11$ or $b\leq -10$.  
However, there continues to be  
debate in the literature about the 
precise range of validity of such models \cite{bm}. 

%There is 
Another aspect of the equations (\ref{bfamily}) that makes them 
the focus of much interest is the special solutions that they admit. 
Although (as already mentioned) there are multi-soliton solutions for $b=2,3$, 
these smooth solutions only exist on a zero background in the case where the 
equation has additional linear dispersion terms (terms proportional to 
$u_x$ and/or $u_{xxx}$, that is); such terms can be removed 
by a combination of a Galilean transformation together with a 
shift $u\to u+u_0$, which (for $u_0\neq 0$) changes the boundary conditions 
at spatial infinity. In the case of vanishing boundary conditions 
at infinity, there are no smooth multi-soliton solutions, but 
Camassa and Holm noticed that for $b=2$ and any positive integer $N$ there are 
instead weak solutions given by  
\beq \label{multipeakon} 
u(x,t) = \sum_{j=1}^N p_j(t)e^{-|x-q_j(t)|}, 
\eeq 
which have the form of a linear superposition of $N$ peaked waves whose 
positions $q_j$ and amplitudes $p_j$ are respectively the canonically conjugate 
coordinates and momenta in a finite-dimensional  
Hamiltonian system that is completely integrable in the Liouville-Arnold sense.
When $b=2$, Hamilton's equations correspond to the geodesic equations 
for an $N$-dimensional manifold with coordinates $q_1,\ldots ,q_N$ and 
co-metric $g^{ij}=e^{-|q_i-q_j|}$.  
The form of the multi-peakon solutions (\ref{multipeakon}) persists for all values 
of $b$, although in general the Hamiltonian system governing the time 
evolution of the positions and amplitudes is non-canonical \cite{hh}, 
and for $N>2$ this finite-dimensional 
dynamics is expected to be integrable 
only when $b=2,3$. 

In the case $b=2$, it is known that the Camassa-Holm equation is of 
Euler-Poincar\'e type, corresponding to 
a geodesic flow with respect to the $H^1$ metric 
on a suitable diffeomorphism group \cite{gert}; the geodesic equations for 
the $N$-peakon solutions
(\ref{multipeakon})  
are a finite-dimensional reduction of this flow \cite{epdiff}.  Although the standard geodesic interpretation, in 
terms of a metric,  is 
lost for other values of $b$, it was recently shown that the periodic case of the Degasperis-Procesi and the other 
equations in the $b$  family can be regarded as geodesic equations for a non-metric connection on the diffeomorphism group 
of the circle \cite{ek}. 

\begin{figure}\centering 
\scalebox{0.25}{ %[0.55]{
\includegraphics[angle=0]{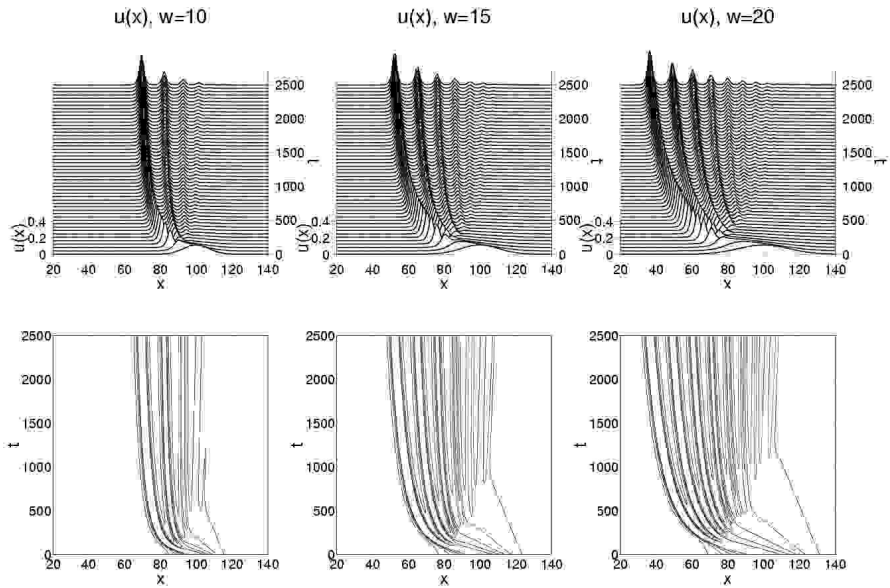}}
\caption{\small{
Leftons evolve from Gaussian initial profiles of different 
widths in the case $b=-3$.  (Reproduced with kind permission from \cite{Holm03}.) 
 }}
\label{numerics} 
\end{figure}

Holm and Staley made an extensive numerical study of solutions of 
(\ref{bfamily}) for different values of $b$, by starting with different 
initial profiles and observing how they evolved with time and with changing $b$ 
\cite{Holm03a, Holm03}. 
They observed that, broadly speaking, there are three distinct parameter regimes 
with quite different behaviour, separated by bifurcations at $b=1$ and $b=-1$, as 
follows:

\vspace{.1in}
\noindent {\bf Peakon regime}: For $b>1$, arbitrary initial data asymptotically separates out into 
a number of peakons as $t\to\infty$. 

\vspace{.05in}
\noindent {\bf Ramp-cliff regime}: For $-1<b<1$, solutions behave asymptotically like a combination 
of a ``ramp''-like solution of Burgers equation (proportional to $x/t$), together with 
an exponentially-decaying tail (``cliff''). 

\vspace{.05in}
\noindent {\bf Lefton regime}: For $b<-1$, arbitrary initial data moves to the left and asymptotically separates out into 
a number of ``leftons'' as $t\to\infty$, which are smooth stationary solitary waves. (See Figure 1.)  

\vspace{.1in}
The behaviour observed separately in each of the parameter ranges $b>1$ and $b<-1$ can be understood as 
particular instances of the {\it soliton resolution conjecture} \cite{tao}, 
a vaguely defined conjecture which states that for suitable dispersive wave equations, 
solutions with ``generic'' initial data will decompose into a finite number of 
solitary waves plus a radiation part which decays to zero.  
In this article, our aim %of our investigation 
is to provide a first step towards explaining 
this phenomenon analytically for the equation (\ref{bfamily}) in 
% explanation for some of these phenomena. we are primarily concerned with 
the ``lefton'' regime $b<-1$. 
We show that  
in this parameter range 
a single lefton solution is orbitally stable, by applying the approach of Grillakis, Shatah and Strauss in \cite{Grillakis87}.   
The main ingredients required for our stability analysis are the Hamiltonian structure and conservation laws for 
(\ref{bfamily}). The lefton solutions are a critical point for a functional which is combination of the Hamiltonian
and a Casimir, but the second variation has some negative spectrum, so it 
is not possible to apply the energy-Casimir method as in \cite{ecmethod}.

In the next section we describe the Hamiltonian structure and conservation laws of  
(\ref{newb}) that exist for all $b$. After that we consider orbital stability of stationary waves when $b<-1$: 
see   
Theorem 2 in section 3 for the main result of the paper.
We make some remarks about other values of $b$  
in our conclusions.

\section{Conserved quantities 
and Hamiltonian structure} 

In order to better understand the properties of each equation  in the family (\ref{bfamily}), 
it is convenient  
to rewrite it in the following way: 
\beq \label{newb} 
m_t +um_x +bu_x m =0, \qquad m=u-u_{xx}. 
\eeq 
This can be regarded as a nonlocal evolution equation for $m$, 
where (at each time $t$) the field $u$ is obtained from $m$ 
by the convolution 
\begin{equation}
\label{g}
u(x) = g*m (x)= \int_{-\infty}^\infty g(x-y) m(y) dy, \qquad 
g(x)=\frac{1}{2}\exp(-|x|).  
\end{equation}
Henceforth we use the symbol $\int$ without limits to  
denote an integral $\int_{\R}$ over the whole real line.

From  the equation (\ref{bfamily}) written in the nonlocal form (\ref{newb})  it is straightforward to verify 
that, for any value of $b\neq 0,1$, there are at least three different functionals 
that are formally conserved by the 
time evolution of $m$ \cite{dhh2}, namely %the Hamiltonian 
{\beq 
\label{ham} 
E=\int %_{\cb\mathbb{R}} 
m\, dx, \qquad  C_1 = \int %_{\cb\mathbb{R}} 
m^{1/b}\, dx, 
\eeq
and 
\beq\label{casimirs} 
C_2 = \int %_{\cb\mathbb{R}} 
m^{-1/b}\left(\frac{m_x^2}{b^2m^2}+1\right)\, dx.
\eeq} 
In saying that these quantities,  each of which has the form 
$\int  \mathcal{T}\, dx$ for some density $\mathcal{T}$, are formally conserved, 
we mean %only 
that there is a flux $\mathcal{F}$ such that the conservation law 
$ 
\frac{\partial \mathcal{T}}{\partial t} = \frac{\partial \mathcal{F}}{\partial x}  
$ 
holds for any smooth solution of the equation (\ref{bfamily}).  
If the integral $\int \mathcal{T}\, dx = \int_{\R} \mathcal{T}\, dx$ exists, and the flux $\mathcal{F}$ vanishes 
at infinity, then clearly $d/dt \int  \mathcal{T}\, dx=0$ 
for strong solutions that decay sufficiently fast 
at infinity.

The smooth solutions of (\ref{bfamily}) can be derived from a variational principle $\delta S=0$, by 
starting from the conservation law $(m^{1/b})_t+(um^{1/b})_x=0$ associated with $C_1$ 
and introducing a potential $\varphi$ such that $\varphi_x =m^{1/b}$, $\varphi_t=-um^{1/b}$. The action is 
$$ 
S=\int\int_{\mathbb{R}^2} \mathcal{L}\, dx\, dt, \qquad 
\mathcal{L} = \frac{\varphi_t}{2\varphi_x} \Big( (\log \varphi_x )_{xx}+1\Big) -\frac{\varphi_x^b}{b-1} 
$$ 
(where we have changed a sign compared with  \cite{dhh2}). After rearranging, 
the Euler-Lagrange equation gives 
$$ 
\frac{\partial}{\partial x}\left( 
\frac{\varphi_t}{\varphi_x} -\Big(\frac{\varphi_t}{\varphi_x}\Big)_{xx}+\varphi_x^b\right) =0, 
$$
which (up to an integration with respect to $x$) is equivalent to (\ref{newb}). 
Noether's theorem applied 
to the time translation symmetry $t\to t+s$ leads to the conserved density 
$$ 
\varphi_t\, \frac{\partial\mathcal{L}}{\partial\varphi_t}- \mathcal{L}=   \frac{\varphi_x^b}{b-1}, 
$$ which (up to scaling) corresponds to $E$ above; the space translation $x\to x+s$ 
leads to an equivalent density. Applying Noether's theorem to the symmetry of shifting the potential
$\varphi \to \varphi +s$ gives the density 
$$ 
\frac{\partial\mathcal{L}}{\partial\varphi_t}=\frac{1}{2\varphi_x}\Big( (\log \varphi_x )_{xx}+1\Big), 
$$ 
which corresponds to $C_2$. The same action $S$ is also valid for $b=0$, but needs to be 
modified slightly for $b=1$.

There is another type of conservation law which holds for 
solutions of  (\ref{newb}), which is the fact 
that 
\beq \label{diffeo} 
m(q,t) q_x^b=m(x,0)
\eeq 
for all $t$ in the domain of existence, where 
$x \mapsto q(x,t)$ is a diffeomorphism of the line defined from the solution 
of the 
initial value problem 
\beq \label{ivp} 
q_t =u(q,t), \qquad  
%with initial condition $ 
q(x,0)=x.   
\eeq 
%%% 
(See \cite{zhou}, 
and also Proposition 9 in \cite{ek} for the case 
of the circle.) 
By adapting McKean's argument  
for the case $b=2$ \cite{mckean}, this implies that if 
a %smooth 
solution is initially positive, %$m(x,0)>0$ , 
then $m(x,t)>0$  everywhere as long as the solution exists. 
In the next section we shall restate %make use of 
a stonger result along these lines in the  case $m\in H^1$, 
which is proved by Zhou in \cite{zhou}.

The choice of nomenclature for the above functionals comes from the fact 
that, for any $b$, the skew-symmetric operator 
\beq\label{bop}
B=-(bmD_x+m_x)(D_x-D_x^3)^{-1}(bD_x m-m_x) 
\eeq   
is  a Hamiltonian  operator 
\cite{wanghone, hh}, 
in the sense that it defines a Poisson bracket 
$$ 
\{ \,F,G\, \} = \left<\frac{\delta F}{\delta m}, B \,\frac{\delta G}{\delta m}\right>, 
$$ 
between any pair of smooth functionals $F,G$, where 
$<f,g>=\int %_\R 
fg\, dx$ denotes the usual %$L^2$ scalar product 
pairing between real functions 
on the line.  
Note that in (\ref{bop}) and 
elsewhere we  use $D_x$ to mean differentiation with respect to $x$.
For suitable functions $f$ %\in \mathrm{im} \,D_x$, 
the inverse operator 
in (\ref{bop}) is defined by  $(D_x-D_x^3)^{-1} f= G*f$, 
taking the convolution 
with $G(x)=\frac{1}{2}\mathrm{sgn}(x)(1-\exp(-|x|))$. %\cite{hh}. 

The quantities $C_1$ and $C_2$ are the Casimirs for this bracket, 
satisfying $\{\, F,C_j\,\}=0$ for any $F$, for $j=1,2$. 
For any $b\neq 1$, the equation (\ref{newb}) can be written in 
Hamiltonian form as 
\beq\label{hamform} 
m_t = \frac{1}{b-1}\, B\, \frac{\delta E}{\delta m}, 
\eeq
with $E$ as in (\ref{ham}) being the Hamiltonian (up to scale); 
for $b=1$ (when $E=C_1$) one should take $\int %_{\cb\mathbb{R}} 
m \log m \, dx$ as the Hamiltonian.

In fact, depending on the solutions considered, one or more of these functionals 
may not be defined. For example, in the case of the leftons, which are the solutions of interest here, we have that $b<-1$ 
and 
$m$ is smooth, rapidly decaying and everywhere positive, so that
$E$ and $C_2$ both exist while $C_1$ does not. 
In the case where $b$ is positive, on the other hand, 
if $m$ is sufficiently  smooth, rapidly decaying  
and positive, %non-negative would be OK here    
then we would have that only $E$ and $C_1$ exist, and $C_2$ does not. 
In the case of a single peakon given by $u=c\exp (-|x-ct|)$, 
the field $m$ is given 
by a delta function, $m=2c\delta(x-ct)$, and similarly for the multi-peakon solution 
(\ref{multipeakon}) it is 
$m=2%\sum_{j=1}^N 
\sum_jp_j(t)\delta(x-q_j(t))$,  
so that the functional $E$ makes sense, but $C_1$ and $C_2$ do not. 

It is worth mentioning that for the integrable cases of (\ref{bfamily}),  the nonlocal Hamiltonian operator (\ref{bop}) 
%{\cb 
defines 
just one of a set of compatible Hamiltonian 
structures. When $b=2$ the first Hamiltonian operator is given by 
$B_1=D_x(1-D_x^2)$, and the second is  
$B_2=mD_x+D_xm$, where the latter 
% the second Hamiltonian structure
defines the Lie-Poisson bracket (the dual of the Euler-Poincar\'e structure); in that case,  the nonlocal operator   
(\ref{bop}) %{\cb 
defines 
the third Hamiltonian structure, %{\cb 
and is given by $B_2B_1^{-1}B_2$  up to scaling. For $b=3$, 
$B$ in (\ref{bop}) %{\cb 
defines 
the second Hamiltonian structure, while 
$B_1=D_x(1-D_x^2)(4-D_x^2)$ is 
the first Hamiltonian operator   
\cite{dhh}.

\section{Stability of the stationary solution} 

In what follows we shall primarily be interested in the ``lefton'' solutions. 
A single lefton is a stationary solution of (\ref{bfamily}) given by the 
explicit formula \cite{dhh2}
\beq\label{lefton} 
u = A\,\Big(\cosh \gamma (x-x_0)\Big)^{-\frac{1}{\gamma}}, \qquad \gamma = -\frac{b+1}{2}  
\eeq 
(independent of $t$), where the position $x_0$ and amplitude $A$ are arbitrary 
constants. 
For $b<-1$, corresponding to $\gamma >0$, when $A>0$ this is a positive, smooth solution 
decaying like $e^{-|x|}$ as $|x|\to\infty$; so asymptotically it has the same shape 
as a peakon solution.   From (\ref{newb}), stationary solutions satisfy $u^bm=$constant.

\subsection{Overview of the theory}

Following \cite{Grillakis87}, we consider  orbital stability, 
which means nonlinear stability for solutions of Hamiltonian systems 
up to drifts along the action of Hamiltonian symmetries.
Suppose that a system in  Hamiltonian form is defined
on a real Hilbert space $\mathcal{X}$, with energy functional (Hamiltonian) $E$, 
and admits a one-parameter Lie group of %point 
Hamiltonian symmetries  
$T_s: \mathcal{X}\rightarrow \mathcal{X}$ (where $s$ is the parameter of the group), 
with infinitesimal generator $T'_0$, where $T_s$ is a unitary operator 
on $\cal X$. 
The Hamiltonian system is %have the form 
\beq \label{hams} 
w_t = J\, \frac{\delta E}{\delta w}, 
\eeq 
with the (skew-symmetric)  Hamiltonian operator $J: \, \mathcal{X}^* \to \mathcal{X}$,  
but one considers weak solutions in $\mathcal{X}$, namely $w$ which satisfy 
\beq\label{weak} 
\frac{d}{dt}\left< \psi , w\right> = -\left< \frac{\delta E}{\delta w}, J\, \psi\right>, 
\eeq 
for all $\psi \in D(J)\subset \mathcal{X}^*$, where $<,>$ denotes the 
pairing between $\mathcal{X}$ and $\mathcal{X}^*$. 
The natural isomorphism $I: \mathcal{X}\to \mathcal{X}^*$ is defined by 
$<Iu,v>=(u,v)$, where $(,)$ is the inner product on $\mathcal{X}$. 

Then one considers  the stability of particular solutions, 
which physically correspond to {\sl{bound states}} 
or {\sl{solitary waves}},
for which $w(t)$ takes the form
\begin{equation}
\label{solform}
%w(t)=
T_{\omega t}\, \phi ,
\end{equation}
for some fixed $\phi\in \mathcal{X}$, depending on the 
parameter $\omega\in\R$, 
which is a critical point  of the 
functional
\begin{equation}
\label{fnal} 
F=E-\omega Q.
\end{equation}
%where
The conserved functional   $Q$ %is another 
(often identified as the {\sl{charge}}) arises from the symmetry via an %the 
%\blue{an 
infinite-dimensional version of 
Noether's theorem \cite{Olver}; this means that the Hamiltonian vector field 
associated with $Q$ generates the 
symmetry $T_s$, in the sense that 
$w_s = J\, \frac{\delta Q}{\delta w}\equiv T_0'\, w$.
%It is assumed that 
Both $E$ and $Q$ are invariant 
under the symmetry group.

The three main assumptions in \cite{Grillakis87} can be paraphrased thus: 
% we are paraphrasing/summarizing? 
\begin{itemize}
\item[(i)] Local existence: For each $w_0 \in \mathcal{X}$ the solution of (\ref{weak}) 
with initial data $w(0)=w_0$
exists in some time interval $t\in [0,t_0 )$, for some $t_0 >0$, and both $E$ and $Q$ are conserved: 
$E(w(t))=E(w_0)$, $Q(w(t))=Q(w_0)$ for all $t$ in this interval. 
\item[(ii)] Existence of bound states: There is a smooth map $\omega \mapsto \phi=\phi_\omega$ 
from some %open 
parameter interval $(\omega_1,\omega_2 )$  
%of values of the real  $\omega $
into $\mathcal{X}$, such that $\phi$ is a critical point for the 
functional $F$ in (\ref{fnal}), i.e. $\frac{\delta E}{\delta w} (\phi ) - \omega \frac{\delta Q}{\delta w} (\phi )=0$; 
also $\phi\in D((T'_0)^2)$
and $T'_0\phi\neq 0$. 
\item[(iii)] For each $\omega\in(\omega_1,\omega_2 )$,   the second variation 
$\frac{\delta^2 F}{\partial w^2}(\phi)$ has exactly one negative simple eigenvalue, has 
its kernel spanned by $T'_0\phi$, and the rest of its spectrum is positive and bounded away from zero. 
\end{itemize}

In the above, of central importance is the second variation of $F$, 
$${\tt H}:=\frac{\delta^2 F}{\partial w^2}(\phi) = 
\frac{\delta^2 E}{\partial w^2}(\phi)-\omega \frac{\delta^2 Q}{\partial w^2}(\phi), 
$$ 
which is a self-adjoint operator from $\mathcal{X}$ to $\mathcal{X}^*$. Its spectrum is defined to be the set of 
$\la\in\R$ such that  ${\tt H}-\la I$ is not invertible.  
Evaluation of the functional $F$ at $\phi$ defines a function $d(\omega )$. 
The solution (\ref{solform}) is 
the $\phi$-orbit $\{ T_{\omega t}\phi \, | \, t\in\R\}$, and its stability is defined in terms of the norm $\| \cdot \|$ on $\cal X$, as follows. 

\begin{Def}\label{origdef} 
The $\phi$-orbit is {\sl{%orbitally 
stable}} if for all $\epsilon >0$ there exists $\delta >0$ 
with the following property. If $w(t)$ is a solution to (\ref{hams}) in some time interval $[0,t_0)$, 
such that $\| w(0)-\phi\| <\delta$, 
then $w(t)$ is defined for $0\leq t<\infty$ and 
$$
\sup_{0<t<\infty} \inf_{s\in \mathbb{R}} \| w(t)-T_s\phi \|  <\epsilon.
$$
\end{Def}

The above definition can be modified in the case where the solution 
$w(t)$ may exhibit blow up in finite time, but this will not be 
needed for our purposes.
One of the main results of \cite{Grillakis87} is the following.

\begin{Thm}\label{gss} 
Given assumptions (i)-(iii) above, for  $\omega\in(\omega_1,\omega_2 )$ the $\phi$-orbit is stable 
if the function $d$ satisfies $d''(\omega)>0$. 
\end{Thm} 

As stated above, 
the aforementioned three sets of assumptions 
are enough to obtain a sufficient condition for stability.
In order to make the stronger statement that the convexity of the function 
$d(\omega )$ is both necessary and sufficient %condition 
for stability, as in Theorem 2 of \cite{Grillakis87}, 
a slightly more stringent version of the assumptions in (ii) 
and a fourth condition on the Hamiltonian operator are required. (See the remark made on p.167 of \cite{Grillakis87}.) In this paper, we will only require the sufficient condition for the stability of the leftons.

Grillakis et al. also explain how (with minor alterations) their approach 
is valid for solutions defined in a Banach space. % ${\cal X}$. 
In the rest of this section, we consider appropriate modifications of their approach for the lefton solutions (\ref{lefton}).  
In particular, it will be necessary to consider stability in 
a certain Banach subspace of a particular Hilbert space, and 
we will discuss below 
how the approach of  \cite{Grillakis87} applies in this case.

\subsection{Choice of a suitable Banach space}

For the lefton solution (\ref{lefton}), the corresponding 
field $m$ has the form $m=m_0(x)$, where 
\beq\label{mlefton} 
m_0 = A \frac{1-b}{2}
\Big( \cosh \gamma (x-x_0) \Big)^{\frac{b}{\gamma}}, 
\qquad \mathrm{with} \quad \gamma = -\frac{b+1}{2}>0  
\eeq 
for $b<-1$,  which is positive and  smooth, and decays like $e^{-|bx|}$ at infinity. 

To begin with, we need to show that $m_0$ is a critical point of a functional defined on an appropriate % suitable 
space. 
As noted above, 
on the real line 
the functional $C_1$ diverges for $b$ negative, so we want to realize the solution above as the critical point 
for a specific linear combination of the functionals $E$ and $C_2$. The main technical difficulty in this case concerns 
the functional $C_2$, which is only defined for certain positive 
(or non-negative) $m$. We can redefine $C_2$ for negative $m$ 
by introducing modulus signs, but the integrand will not be smooth wherever $m$ has a zero; for this reason 
we would like to consider solutions of (\ref{newb}) with $m>0$ everywhere.

The crucial observation to make is the fact that 
the solution (\ref{mlefton}) 
satisfies the first order 
differential equation 
\beq \label{mfirst} 
m_x^2=b^2\left(m^2-\frac{m^{3+1/b}}{k}\right).
\eeq 
From this it follows that  $m=m_0$  is a critical point for the functional
\begin{equation}
\label{functdef}
F = -E + k\, C_2, %\qquad k = \left(A\frac{1-b}{2}\right)^{1+1/b}, 
\end{equation}
corresponding to the value 
\beq\label{kval} 
k = \left(A\frac{1-b}{2}\right)^{1+1/b}. 
\eeq 
Indeed, for suitable smooth $m,v$ such that $m+\epsilon v$ is positive and $F$ is defined there whenever $|\epsilon|$ is small enough, 
%{\cv{then as long as $v$ converge to zero fast enough as $|x|\rightarrow\infty$}}, 
the first variation  in the direction $v$ is 
$\delta F(m)\, v =\lim_{\epsilon\to 0} \frac{d}{d\epsilon} F(m+\epsilon v)=<\frac{\delta F}{\delta m}, v>$, hence 
\begin{equation}
\label{1strel}
\begin{array}{lll}
\displaystyle{\frac{\delta F}{\delta m}} & \displaystyle{=} & \displaystyle{-\frac{\delta E}{\delta m} + k\, \frac{\delta C_2}{\delta m}} \\ 
&\displaystyle{ =} & 
\displaystyle{-1 + k m^{-1/b-1}  \left( 
\frac{(1+2b)}{b^3} \frac{m_x^2}{ m^{2}} -\frac{2}{b^2}\frac{ m_{xx}}{m}- \frac{1}{b} 
\right).}
\end{array} 
\end{equation}
The latter expression vanishes for the lefton solution: %$m=m_0$: 
$\frac{\delta F}{\delta m} (m_0) = 0$ for $m_0$ given by (\ref{mlefton}) 
and $k$ given by (\ref{kval}).

In order to apply the results of  \cite{Grillakis87} we must restrict to a suitable 
space  in which the functionals $E$ and $C_2$ are twice differentiable, at least near to $m_0$. 
To do so, we first introduce the weight
\beq \label{wt}
\al  
:= %\equiv 
m_0^{-2-1/b}=\left(A \frac{1-b}{2}\right)^{-2-1/b} \, \Big(\cosh \gamma (x-x_0)\Big)^{-\frac{2b+1}{\gamma}},
\eeq 
and consider the equation (\ref{newb}) defined in the space 
$L^2_\al := L^2(\R ,\al\, dx)$. With the standard pairing $<,>$, this gives the isomorphism 
%$I:\, 
$u\mapsto \al u$ from 
$L^2_\al$ to its dual. 
%$\mathcal{X}=L^2_\al$ to $\mathcal{X}^*$. 
The reason for 
this choice of the weight $\alpha$ will become clear shortly when we consider the second variation of $F$.

The second variation of $E$ is zero, so the entire contribution to the second variation of $F$ comes from $C_2$. 
Assuming that $C_2(m+\epsilon v)$ is defined for suitably smooth $m,v$, we have 
\beq \label{2nd} 
\delta^2 C_2 (m, v) :=   
\lim_{\epsilon\to 0} \frac{d^2}{d\epsilon^2} F(m+\epsilon v) 
= \int (\mathrm{P}\, v_x^2 + \mathrm{Q}\, v^2)\, dx,
\eeq  
after performing an integration by parts, where 
$
\mathrm{P}=2m^{-2-1/b} / b^2$, 
and  
$$ 
\mathrm{Q}=\frac{m^{-4-1/b}}{b^4}\Big(2b(1+2b)m_{xx} m-(1+2b)(1+3b)m_x^2+b^2(1+b)m^2\Big).
$$
Evaluating this at $m=m_0$ and using (\ref{mfirst}) gives 
$$ 
\delta^2 C_2 (m_0, v) 
 = \frac{2}{b^2}\int %_{\cb\mathbb{R}} 
\Big(\al \, v_x^2 - b(b+1)\al \, v^2\Big)\, dx, 
$$ 
with $\al$ as in (\ref{wt}). 
To ensure that the second variation of $C_2$  in the direction $v$  is defined at $m_0$, we require 
that $v$ belongs to the Hilbert space  $H^1_\al:=H^1(\R ,\al\, dx)$  
which has the  inner product 
$$ 
(v,w)_\al = \int (vw+v_xw_x) \, \al\, dx,
$$ 
and corresponding norm 
$|| \, v \, ||_\alpha=\sqrt{(v,v)_\al}$,  
so that $ |\delta^2 C_2 (m_0, v)|
\leq K\, || v ||_\alpha^2 $ 
for a universal constant $K$.  
Hence we should consider solutions of (\ref{newb}) 
with  $m\in H^1_\al$. 
Since $\al$ is uniformly bounded away from zero, 
$H^1_\al$ is a subspace of $H^1$. 

However, as we shall explain  further below, 
there is no neighbourhood of $m_0$ in $H^1_\al$ where $C_2$ 
exists, which leads us to consider 
a subspace $\mathcal{Z}\subset H^1_\al$, defined by 
\begin{equation}
\label{Z}
\mathcal{Z} := \{\, f\in H^1_\al \, | \, f=O(m_0) \quad \mathrm{as}\quad |x|\to\infty \, \}. 
\end{equation}
Since functions in $H^1$ are continuous, by the Sobolev embedding theorem, it follows 
that $f\in\mathcal{Z}$ is a continuous function, so there exists some $\K\geq 0$ such that 
$|f(x)|\leq \K \, m_0(x)$ for all $x\in\R$. 
Then for any such $f$ one can define 
\beq\label{knorm}
\K_f:=\sup_{x\in\R}\frac{|f(x)|}{m_0(x)}. 
\eeq 
With this definition, it is easy to check that ${\cal Z}$ is a Banach space with respect to the norm 
\begin{equation}
\label{Zn}
|| f||_{\cal Z}:=||f||_\alpha + \K_f. 
\end{equation}

\subsection{Definition of stability and verification of assumptions}

Henceforth we are going to consider the orbital stability of the 
solution $m_0$ in the Banach space $\cal Z$, 
allowing for translations in the independent variable 
$x$, 
so we begin with a precise definition of this, analogous to  
Definition 1. One of the requirements in \cite{Grillakis87}  
is that $T_s$ should be a unitary operator; but 
the norm $\| \cdot \|_{\cal Z}$ is not invariant under translations, 
while the $H^1$ norm is, so this is involved in the 
definition of stability adopted here.

\begin{Def} \label{stadef}
The solution $m_0$  is {\sl{stable}} if for all $\epsilon >0$ there exists $\delta >0$ 
with the following property. If there is a solution $m(\cdot ,t)$ to (\ref{newb}) in some time interval $[0,t_0)$, 
such that $\| m(\cdot ,0)-m_0 \|_{\mathcal{Z}} <\delta$, 
then $m$ is defined for $0\leq t<\infty$ and 
$$
\sup_{0<t<\infty} \inf_{\xi\in \mathbb{R}}  
\| m(\cdot ,t)-m_0(\cdot - \xi )\|_{H^1}  <\epsilon.
$$
\end{Def}
Our main result is the following.  %theorem.

\begin{Thm}
\label{thethe}
The lefton solution (\ref{mlefton}) is stable in the sense of Definition \ref{stadef}.
\end{Thm}

In order to prove the above result, we now explain how
the assumptions of \cite{Grillakis87} hold, up to 
appropriate modifications. 
While our proof uses the tools developed in \cite{Grillakis87}, our presentation follows that of \cite{cstrauss3} quite closely. 
Here the relevant one-parameter symmetry group is spatial translation,
\beq \label{group} 
T_s \, m(x,t) = m(x+s,t), 
\eeq 
which has the infinitesimal generator $T_0'=D_x$. Both $E$ and $C_2$ are invariant under this symmetry, 
which commutes with the time evolution, i.e. $[D_t,D_x]=0$. We wish to show the stability of the 
$m_0$-orbit for this group. However, unlike the charge $Q$, the functional $C_2$ cannot generate 
the flow of $m_0$ along the orbit, because it is a Casimir. Nevertheless, the fact that $m_0$ is 
a critical point of the functional (\ref{functdef}), with nontrivial dependence on the parameter $k$, 
is sufficient for the methods in \cite{Grillakis87} to work. Note that, since $b$ is negative, compared with 
(\ref{fnal}) we have taken $E\to -E$ in (\ref{functdef}), and for the 
Hamiltonian operator we have $J\to (1-b)^{-1} B$.

Henceforth we assume that $k$ is a free parameter while $A$ is specified by the relation
(\ref{kval}); 
this is to simplify the notation for the stability analysis. 
Furthermore, we assume  %without loss of generality,
that $k$ is positive. 
To consider positive solutions $m$,  
we require $A>0$, which follows 
from %(\ref{kval}) with 
the inverse formula $A(k)=2 k^{b/(b+1)}/(1-b)$ with $k>0$. 

We now consider the three main assumptions from subsection 3.1 in more detail. 
The first of these assumptions concerns local existence and conservation laws.
 
To begin with we briefly discuss local existence of solutions, 
which is the first part of assumption (i) above. 
There are several papers which prove results on 
local existence and blow up of the solutions of (\ref{bfamily}), 
either for particular values of $b$, e.g. in \cite{ce} for $b=2$, 
and in \cite{yin} for $b=3$, or for the whole family 
of equations on the line \cite{zhou}; another  
family of equations that includes the case $b=2$ is treated in \cite{zhourod}. 
Analytic solutions of (\ref{bfamily}) 
are considered in \cite{coclite}, while the equation 
with additional linear dispersion 
(terms proportional to $u_x$ and $u_{xxx}$) is treated in 
\cite{kodzha} and \cite{liu}; for the periodic case, 
see also \cite{christov, katelyn}.  
Following the 
approach for the Camassa-Holm equation in \cite{ce}, one can rewrite 
(\ref{newb}) as a %abstract
quasi-linear evolution equation in $L^2$, that is 
$$ 
m_t + {\tt A} (m)\, m =0, 
$$ 
with the operator ${\tt A} (m) = (g*m)\, D_x + b(g*m)_x \, \mathrm{id}$, 
where  ${\tt A}(m)\in L(H^1,L^2)$ for $m\in H^1$. This way of 
presenting the equation allows the application of Kato's theorem, which gives local well-posedness for $m\in H^1$. 
Alternatively, one can use 
the local existence result 
for (\ref{bfamily}) with solutions $u\in H^s$ for $s>3/2$ 
as stated by Zhou (see \cite{zhou}, Theorem 1.1), which 
is equivalent to $m\in H^{s-2}$; 
so for $u\in H^3 $ this also gives local existence of solutions 
$m\in H^1$. 

Although the preceding results could be modified 
to give a local existence result for $m\in \cal Z$, this 
is not necessary for our purposes. Instead, we will show 
that taking suitable initial data $m(\cdot ,0)\in {\cal Z}\subset H^1$ 
gives global existence of the solution in $\cal Z$, 
as in Theorem \ref{glob} below, 
and the arguments leading up to this only require local existence 
in $H^1$.     
 
The other part of assumption (i) concerns the conservation laws.  
\begin{Lem} \label{energy}
Suppose the initial data $m(\cdot ,0)\in H^1 \cap L^1$ 
is everywhere non-negative. %does not change sign. 
Then the energy $E$ is constant, with 
$$ 
E=\| m \|_{L^1}<\infty 
$$ 
as long as the solution %$m(x,t)$ 
of  (\ref{newb}) %with initial data $m(x,0)$ 
exists. 
\end{Lem}
\begin{proof} 
The fact that $E=\int m\, dx$ is conserved follows 
immediately upon noting that (\ref{newb}) 
takes the Hamiltonian form (\ref{hamform}), 
and for non-negative solutions $E$ is the same 
as the $L^1$ norm of $m$.    
(For a direct proof that does not use the Hamiltonian property, 
the  proof of Lemma 3.4 in \cite{ce}, 
for the case $b=2$, can be adapted to all values of $b$.)  
\end{proof}  

Now $m\in H^1_\al $ implies that $m=o(\al^{-1/2})$ for large $|x|$, 
hence $m\in H^1 \cap L^1$,  
so in 
particular the preceding lemma applies 
to initial data in ${\cal Z}\subset H^1$, as does the next result 
on global existence in $H^1$.

\begin{Thm} \label{zhou} 
If $m(\cdot ,0)\in H^1\cap L^1$ is positive then the solution to 
(\ref{newb})  exists globally in $H^1$ and remains positive for all $t>0$.  
\end{Thm}
\begin{proof} 
This is essentially just a restatement of Theorem  2.5 in the paper by Zhou \cite{zhou}, but here we 
sketch a different proof.
%Recall that blow-up occurs if and only if the wave breaks, i.e. if $u_x\to -%\infty$ occurs at some finite time. 
Note that from (\ref{g}) we have  $u=g*m$, so that $u_x=g_x*m$, and hence 
\beq \label{uxbd} 
\| u_x\|_{L^\infty} \leq \|g_x\|_{L^\infty}\, \| m\|_{L^1}=E/2, 
\eeq
where we have used Lemma \ref{energy}. 
Then an application of integration by parts together 
with  (\ref{newb}) yields 
$$ 
\begin{array}{rcl} 
\frac{d}{dt}\int m^2\, dx & = & 2\int mm_t\, dx \\ 
& = & -2\int umm_x -2b\int u_xm^2\, dx \\ 
& = & (1-2b)\int u_xm^2\, dx. 
\end{array} 
$$ 
A similar calculation shows that 
$$ 
\frac{d}{dt}\int m_x^2\, dx = -(2b+1)\int u_x m_x^2 \, dx+ b\int u_xm^2 \, dx,    
$$ 
which combines with the previous one to give 
$$ 
\begin{array}{rcl}   
\frac{d}{dt} 
\| m\|_{H^1}^2 
& = & (1-b)\int u_xm^2 \, dx -(2b+1)\int u_xm_x^2\, dx   \\ 
& \leq & \frac{E}{2}\, \mathrm{max} (1-b, -2b-1) \,  \| m\|_{H^1}^2 , 
\end{array} 
$$ 
where we have used (\ref{uxbd}) (and note that both $1-b$ and 
$-2b-1$ are positive for $b<-1$, which is the case of interest here). 
Thus from Gronwall's inequality we see that $\|m\|_{H^1}$ remains 
bounded for all $t>0$.  
 
Now for the diffeomorphism $q(x,t)$ 
defined by (\ref{ivp}), 
it follows that $q_x(x,t)$   is positive for all 
$t$. To be precise,  
\beq 
\label{qxbd} 
0<q_x(x,t) =\exp \Big(\int_0^t u_x(q(x,s),s) \, ds\Big) 
\leq \exp(Et/2)  \eeq 
holds for all $t$,  
so from (\ref{diffeo}) it is clear that for every $t>0$, 
$m(x,t)>0$ holds for all $x$, 
since  %provided %that 
$m(x,0)>0$ everywhere. 
\end{proof} %In the proof that follows, 

To see that the functional $C_2$ is conserved, note that as long as $\frac{\delta C_2}{\delta m}$ is defined 
we have 
$$ 
\frac{dC_2}{dt} = \left<\frac{\delta C_2}{\delta m}, m_t\right> = -\frac{1}{(b-1)}\left<\frac{\delta E}{\delta m}, 
B\, \frac{\delta C_2}{\delta m}\right>=0
$$ 
for a dense class of $m\in H^1$. To be precise, from the coefficient of $k$ on the right hand side of 
(\ref{1strel}) one sees that $B\, \frac{\delta C_2}{\delta m}=0$ for 
sufficiently smooth $m$ - 
this is just the statement that $C_2$ is a Casimir -  and for general 
$m$ one should approximate by smooth functions. 
This means that for initial data in the subspace ${\cal Z}\subset H^1$, 
the value of $C_2$ remains constant, and it turns out that 
if $C_2$ is finite then positive solutions remain in this subspace.

\begin{Thm}\label{glob} Suppose that the 
initial data $m(\cdot ,0)\in \cal Z$ is everywhere positive 
with $C_2<\infty$. Then the solution of (\ref{newb}) exists  
globally in $\cal Z$.  
\end{Thm} 
\begin{proof}
Writing $\K_m(t)$ to denote the supremum part (\ref{knorm}) of 
the $\cal Z$ norm  of $m(\cdot ,t)$, %$\|\|_{\cal Z}$, 
we have 
$$ 
\K_m (t) =   \mathrm{sup}_{x\in\R}\frac{|m(x,t)|}{m_0(x)} 
=\mathrm{sup}_{x\in\R}\frac{|m(q,t)|}{m_0(q)}
=\mathrm{sup}_{x\in\R}\frac{|q_x^{-b}m(x,0)|}{m_0(q)}, 
$$
using the fact that the solution $q=q(x,t)$ of (\ref{ivp}) is a 
diffeomorphism for all $t\geq 0$, 
followed by (\ref{diffeo}). From this it is clear that 
\beq\label{km}
\K_m(t) \leq \mathrm{sup}_{x\in\R}\frac{|q_x^{-b}|m_0(x)}{m_0(q)}
\, \K_m (0). 
\eeq 
To bound this further, consider $\rho(x,t)=m_0(x)/m_0(q)$, 
which is seen to satisfy  
$$
\frac{\partial\rho}{\partial t}= 
-\rho \, q_t\,  \frac{m_0'(q)}{m_0(q)}
= -b\rho\,  u(q,t)\, \mathrm{tanh}\gamma (q-x_0),  
$$ 
upon using (\ref{ivp}) once again, as well 
as the explicit formula (\ref{mlefton}). 
The hyperbolic tangent is bounded above by 1, and 
$u=g*m$ implies 
$\| u\|_{L^\infty}\leq \| g\|_{L^\infty}\| m\|_{L^1}=E/2$, 
so overall this gives 
$$ 
\frac{\partial\rho}{\partial t}\leq -\frac{bE}{2} \rho, 
$$ 
and hence Gronwall's inequality 
(together with $\rho (x,0) =1$) yields
$\rho (x,t) \leq \exp (-bEt/2)$.  Similarly, the term 
$q_x^{-b}$ in (\ref{km}) has an upper bound obtained 
from (\ref{qxbd}), 
so that overall  
$
\K_m(t) \leq \exp (-bEt)\, \K_m(0)
$ for all %\qquad \mathrm{for} %\quad \mathrm{all} 
$t\geq 0$. 
Now rewriting the integrand in (\ref{casimirs}) as 
$(m/m_0)^{-2-1/b}(b^{-2}m_x^2+m^2)\,\al$ gives the inequality 
$ 
\| m\|_\al^2 \leq b^2 \, C_2 \K_m^{2+1/b}<\infty $, 
so $\| m\|_{\cal Z}$ remains bounded for all $t>0$. 
\end{proof}

For the assumption (ii) above, we note that the bound state $\phi = m_0$ is a smooth function of $k$, it satisfies $T_0'\, m_0\neq 0$, and 
we have also shown that $m_0$ is 
a critical point of the functional (\ref{functdef}). The main 
point to discuss is whether the functionals $E$ and $C_2$ (and hence $F$) 
are twice differentiable in a neighbourhood of $m_0$. 
This turns out to require further restrictions on the allowed variations. 
There is no problem with $E$, as 
for $m=m_0+v$ we have 
$$ 
|E(m)-E(m_0)|=\left|\int
%_{\cb\mathbb{R}} 
v\, dx \right|\leq ||v||_{L_1}\leq(\al^{-1}, |v|)_\al\leq ||\al^{-1}||_\al \, ||v||_\al, 
$$ using the $H^1_\al$ inner product $(,)$ followed by Cauchy-Schwarz, which gives continuity of $E$ in 
$H_\al^1$; and $E$ is also clearly a smooth functional.
 For $C_2$, it is worthwhile to consider  piecewise smooth functions $v\in H^1_\al $ which are asymptotic to a multiple of $e^{-c |x|}$ 
as   $|x|\to\infty$, where $c> -b-\frac{1}{2}$ to ensure $||v||_\al <\infty$. For all such functions, $C_2(v)$ 
exists provided $v>0$. However, if  $-b>c$ then there is at least one choice of sign of $\epsilon$ 
such that $m_0(x)+\epsilon v(x) < 0$ for either positive or negative $x$ of large enough magnitude.   
For example, for the family of positive functions $v_c = \exp(-c|x|)$ we have 
$$ 
 C_2(v_c) =-2\left( \frac{b}{c}+\frac{c}{b}\right)  \qquad  \mathrm{and} \qquad 
||v_c||_\al <\infty \qquad \mathrm{for} \qquad c> -b-\frac{1}{2}.
$$ 
However, for all $\epsilon <0$ 
and $c<-b$  
it is clear that $m_0(x) +\epsilon v_c(x) <0$ whenever $|x|$ 
is sufficiently large. This means that there is no neighbourhood of $m_0$ in $H_\al^1$ 
where $C_2$ is a smooth functional. 

To rectify this problem, we consider a neighbourhood of $m_0$ in $\cal Z$. 
\begin{Lem} \label{c2def} 
For all $v\in\mathcal{Z}$ there exists an $R$ 
such that $C_2(m_0+\epsilon v)$ is a smooth function of $\epsilon$ for $|\epsilon |<R$.  
\end{Lem} 
\begin{proof} By replacing $v$ by $\epsilon v$ for suitably small $\epsilon$ if necessary, one 
can assume that $K_v<1$, which implies that $m=m_0+v$ is a positive function, and 
$$ 
\frac{||m||_\al^2 }{b^2(1+\K_v)^{2+\frac{1}{b}}}\leq C_2(m) \leq  \frac{||m||_\al^2 }{(1-\K_v)^{2+\frac{1}{b}}}. 
$$
Then for all $v$ with $\K_v<1$, $C_2(m_0+\epsilon v)$ is defined, and the integrand 
in (\ref{casimirs}) is a bounded differentiable function of $\epsilon$, provided that  $|\epsilon | <1$.
\end{proof} 

In fact the above proof shows that, with respect to the norm  $||\cdot ||_{\cal Z}$, the functional $C_2$ is smooth in 
a ball of radius 1 around $m_0$.  
The preceding considerations make it clear that in order to apply the
results in  \cite{Grillakis87} we should consider initial data in ${\cal Z}$, 
as in Definition 2, 
and all variations must be taken in this subspace of $H^1_\alpha$.  
Having found a suitably restricted class of variations, % in $H^1_\al$, 
we proceed 
to verify assumption (iii) above, 
which concerns the operator 
%We also need to consider the second variation of $F$, which reads
\begin{equation}
\label{op}
{\tt H} \equiv \frac{\delta^2 F}{\delta m^2}(m_0)=%\frac{k}{\al}
k\, \Big(-D_x \mathrm{P}_0 D_x+\mathrm{Q}_0\Big),
\end{equation}
where $\mathrm{P}_0$ and $\mathrm{Q}_0$ are, respectively, the coefficients functions $\mathrm{P}$ and $\mathrm{Q}$ defined after (\ref{2nd})
evaluated at $m=m_0$. 

\begin{Lem} \label{spectrum} 
On 
$H^1_\al$ 
the operator ${\tt H}$  defined in (\ref{op}) 
has only one negative eigenvalue, its kernel is one-dimensional, 
and the positive part of the spectrum 
is bounded below away from zero. 
\end{Lem}
\begin{proof}
Upon evaluating 
the quantities depending on $m$ 
at $m=m_0$, 
the eigenvalue problem $({\tt H} -\lambda \alpha )y=0$ 
associated to ${\tt H}$ can be written as
\begin{equation}
\label{eqnon}
k\left(-D_x\mathrm{P}_0 D_x+\mathrm{Q}_0\right)y=\lambda\alpha y, 
\end{equation}
where $\mathrm{P}_0=2\alpha /b^2$, $\mathrm{Q}_0 = -2(b+1)\alpha /b$. 
We now make the change of variables 
$\widetilde{y}=\sqrt{\alpha}y$,  
so that the eigenvalue problem becomes
\begin{equation}
\label{opc}
{\tt L}\widetilde{y}=\lambda\widetilde{y},\qquad {\tt L}=k\left(-\frac{2}{b^2}D_{x}^2+\widetilde{\mathrm{Q}}\right),  
\end{equation}
with $\widetilde{\mathrm{Q}}=\alpha^{-1/2}(\alpha^{-1/2}\alpha_x)_x/b^2-2(b+1)/b$. 
To find the continuous spectrum of ${\tt L}$, we define
$${\tt L}^{\infty}=\lim_{x\rightarrow \infty}{\tt  L}=k\left(-\frac{2}{b^2}D_{x}^2+\frac{1}{2b^2}\right),$$
where we have used the fact that $\alpha$ grows like $\exp(-(2b+1)|x|)$ as $|x|\to\infty$. 
The continuous spectrum of ${\tt L}$ is then given
by the set
\begin{equation}
\label{Henry}
%S=
\left\{\lambda\in\,\mathbb{C}\,\, | \,\, {\tt L}^\infty(\sigma) = \lambda \;\; {\mbox{for
some}}\;\;\sigma\in\mathbb{R}\right\}
\end{equation}
(see \cite{Henry81}, Theorem A.2, p.\ 140),  
where ${\tt L}^\infty(\sigma)$ is obtained from ${\tt L}^\infty$ by replacing $D_x$ with $i\, \sigma$.  The continuous spectrum 
of ${\tt L}$ in $H^1$ thus consists of the interval $\left[\frac{1}{2}k/b^2,\infty\right)$ (recalling that $k$ is positive). 
The result then carries over to ${\tt H}$ through the change of variables.

The eigenvalue problem associated with ${\tt L}$ is an irregular Sturm-Liouville problem,  
with endpoints $\pm\infty$ both being limit-points. 
Due to this and the fact that 
the continuous spectrum  is bounded below, 
the discrete spectrum below the continuous spectrum consists of simple eigenvalues which are ordered 
according to the number of zeros of the corresponding eigenvector, with no two eigenvectors having the same number of zeros and
with the lowest eigenvalue corresponding to
an eigenvector with no zero (case 8.iii of Theorem 10.12.1 in \cite{Zettl}).  

The kernel of ${\tt L}$ is found by making the observation that $\widetilde{y}=\sqrt{\alpha}m_{0,x}$ solves (\ref{opc}) for $\lambda=0$. Equivalently, 
it is easy to verify that $T_0' \, m_0 = m_{0,x}$ is in the kernel of ${\tt H}$. 
Since the eigenvalue zero is simple, there is nothing else in the kernel of ${\tt L}$. Furthermore, since
$m_{0,x}$ has one zero, there is one and only one negative eigenvalue. By multiplying (\ref{eqnon}) by $y$ and 
integrating over $\mathbb{R}$, it follows that 
the negative eigenvalue is bounded below by $-2k(b+1)/b<0$. As before, these results carry over to ${\tt H}$ through the change of variables. 
\end{proof}

\subsection{Proof of stability}

In order to carry out the proof of Theorem \ref{thethe}, we introduce another notion from \cite{Grillakis87}. 
\begin{Def} 
The tubular neighbourhoods of $m_0$ in $H^1_\alpha$ and $\cal Z$ are given by 
$$ 
U_\epsilon =\{ f\in H^1_\alpha \vert \inf_{s\in\mathbb{R}} \| f(\cdot + s) -m_0\|_\alpha <\epsilon \} 
$$ 
and 
$$ 
U_\epsilon^{\cal Z} =\{ f\in {\cal Z} \vert \inf_{s\in\mathbb{R}} \| f(\cdot + s) -m_0\|_{\cal Z} <\epsilon \} , 
$$ 
respectively. 
\end{Def} 

\begin{Lem} \label{GSSlem} 
There exist $\epsilon>0$ and a $C^1$  map $s:U_\epsilon \rightarrow \mathbb{R}$ such that 
for every $v\in U_\epsilon$,
\begin{equation}
\label{s}
\left(v(\cdot +s(v)),\,m_{0,x}\right)_{\al}=0.
\end{equation}
\end{Lem} 
\begin{proof}
Consider the function $\rho(s)=\left(v(\cdot +s),\,m_{0,x}\right)_{\al}$. We have that $\rho'(s)=\left(v_x(\cdot +s),\,m_{0,x}\right)_{\al}$. Thus, when evaluated at $v=m_0$ and $s=0$, we have $\rho(0)=0$ and  $\rho'(0)=\| m_{0,x}\|^2_{\al}>0$. By the implicit function theorem, there is a ball $B_\epsilon \subset H^1_{\al}$ of radius $\epsilon$ around $m_0$,
 an open interval $\cal I$ around the origin in $\mathbb{R}$, and a $C^1$ map $s:B_\epsilon \rightarrow \cal I$ such that that   
the equation $\rho(s)=0$ has a unique solution $s=s(v)\in \cal I$ for all $v\in B_\epsilon$. The result follows 
by noting that the tubular neighborhood in $H^1_\alpha$ is equivalently defined by  
$U_{\epsilon}=\left\{v(\cdot +r) \vert v\in B_{\epsilon},\;r\in\mathbb{R}\right\}$, and the map $s$ 
extends to the whole of $U_{\epsilon}$ by setting  $s(v(\cdot +r))=s(v)-r$.
\end{proof}

Now we define the scalar function
\begin{equation}
\label{d}
d(k)=F(m_0)\equiv -E(m_0) + k\, C_2(m_0), 
\end{equation}
where the lefton $m_0$, as in (\ref{mlefton}), depends on $k$ via $A(k)=2 k^{b/(b+1)}/(1-b)$.  
\begin{Lem}
\label{poslem}
Suppose that $d''(k)>0$. Then there exists a constant $\zeta >0$
such that if $y\in H^1_\alpha$ and $\left< C_2'(m_0),\,y \right>=0=\left(m_{0,x},y\right)_{\al}$, then
 \begin{equation}
 \label{ine}
 \left< {\tt H} y,\,y\right>\geq \zeta \|y\|^2_{H^1}.
 \end{equation}
 \end{Lem}
 \begin{proof} 
Using primes to denote variational derivatives $\delta /\delta m$, we 
differentiate the relation  
$F'(m_0) = - E'(m_0) + kC_2'(m_0)=0$ with respect to $k$, to find  
$C_2'(m_0)=-{\tt H}m_{0,k}$.
Furthermore, $d'(k)=C_2(m_0)$ and thus 
$d''(k)=\left<C_2'(m_0),\,m_{0,k}\right>=-\left<{\tt H}m_{0,k},\,m_{0,k}\right> >0$. 

We now consider the spectral decomposition with respect to the eigenvalue problem 
$({\tt H}-\lambda I )y=0$, where $I=\alpha - D_x \alpha D_x$ is the isomorphism 
from $H^1_\alpha$ to its dual; the properties of the associated spectrum are as
described in Lemma \ref{spectrum} (even if the non-zero eigenvalues and 
corresponding eigenvectors are different 
from those of the problem  $({\tt H}-\lambda \alpha )y=0$ treated in its proof).  
Letting $\chi$ denote the negative eigenvector, such that ${\tt H}\chi=-\mu^2 I \chi$, with 
$\| \chi \|_\alpha =1$, we expand $m_{0,k}=a_0\chi+b_0 m_{0,x}+p_0$ for some $p_0\in P$, where $P$ is the positive subspace for 
${\tt H}$, and $a_0$ and $b_0$ are constants.  Then 
$\left<{\tt H}m_{0k},\,m_{0k}\right><0$, as above, implies $\left<{\tt H}p_0,p_0\right> < a_0^2\mu^2$. 
 
Next take $y$ belonging to the subspace ${\cal S}\subset H^1_\alpha$ defined by the pair of conditions 
$\left< C_2'(m_0),\,y \right>=0=\left(m_{0,x},y\right)_{\al}=0$. On the one hand, 
by the second condition, every such $y$ has the unique representation
$y=a\chi +p$ for some $p\in P$ and constant $a$. The first condition then yields  
$$
 0=-\left< C_2'(m_0),\,y \right>=\left< {\tt H}m_{0,k},\,y \right>=-a_0 a \mu^2+\left<{\tt H}p_0,\,p\right>, 
$$
and a direct calculation (as in the proof of Theorem 3.3 in \cite{Grillakis87})  
shows that $\left<{\tt H}y,\,y\right> >0$ 
for all non-zero $y$.  
On the other hand, observe that there is the direct sum decomposition 
${\cal S}=\tilde{P} \oplus \mathrm{Span} \{ \psi \}$, where $\tilde{P}$ is the subspace 
consisting of all $\tilde{p}\in P$ such that $\left< {\tt H}p_0,\tilde{p} \right> = 0$, 
and $\psi  = \tilde{a}\chi +p_0$ with $\tilde{a} = \frac{\left<{\tt H}p_0,p_0\right>}{a_0\mu^2}$. 
Moreover, this is an orthogonal direct sum with respect to the bilinear form defined by ${\tt H}$. 
%$ \left<{\tt H} y,\tilde{y}\right>$. 
This form is positive definite on  ${\cal S}$ 
and coercive on $\tilde{P}$, since by Lemma \ref{spectrum} there exists some 
$\tilde{\zeta}>0$ such that $ \left<{\tt H} p,p\right> > \tilde{\zeta}\| p\|_\alpha^2$ for all 
$p\in P$. Upon writing any $y\in\cal S$ as $y = \tilde{p} + \tau \psi $ for some constant $\tau$, 
$ \left<{\tt H} y,y\right> =  \left<{\tt H} \tilde{p},\tilde{p}\right>+\tau^2 \left<{\tt H} \psi ,\psi \right>$ 
and $\| y\|_\alpha^2 \leq 2 \| \tilde{p}\|_\alpha^2 + 2\tau^2 \| \psi \|_\alpha^2$ together 
imply $ \left<{\tt H} y,y\right> \geq \hat{\zeta} \| y \|_\alpha^2$, where 
$\hat{\zeta} =\frac{1}{2} \mathrm{min}(\tilde{\zeta}, \| \psi \|_\alpha^2/\left<{\tt H} \psi ,\psi \right>)$. 
The inequality (\ref{ine}) follows by noting that
\beq 
\| y \|_{H^1}^2 \leq k^{(2b+1)/(b+1)}\| y \|_\alpha^2\label{kbd} 
\eeq  
for all $y\in H^1_\alpha$.
\end{proof}
 
We are now ready to prove Theorem \ref{thethe}.
%\begin{proof}
First of all, to verify that $d''(k)>0$, note that 
using (\ref{mfirst})  in (\ref{d}) yields 
$$
d(k) =  2\int%_{\cb\mathbb{R}} 
(km_0^{-1/b}-m_0)\, dx, 
$$  
and this is just proportional to $k^{b/(b+1)}$.  
Then we compute %that
\begin{equation}
\label{dpp}
d''(k)=\frac{-2\hat{K}bk^{-1-1/(b+1)}}{(b+1)^2}, %\int_{-\infty}^{\infty}{m_0\,dx}>0, 
\end{equation}
where 
$$ 
\hat{K}=\int%_{\cb\mathbb{R}} 
\Big({\mbox{sech}}\gamma x\Big)^{1/\gamma}\tanh^2 \gamma x\, dx >0, 
%\;\;{\mbox{with}}\;\;\gamma = -\frac{b+1}{2}.
$$ 
with $\gamma = -(b+1)/2>0$ (recalling that $b<-1$ for the leftons). 
We conclude that the quantity in (\ref{dpp}) is positive,  
and thus Lemma \ref{poslem} is applicable. 

We now show that there exists $\epsilon>0$ such that
\begin{equation}
\label{basic}
E(m_0)-E(m)\geq \frac{\zeta}{4}\|m(\cdot +s(m))-m_0 \|_{H^1}^2
\end{equation}
for all $m\in U^{\mathcal{Z}}_\epsilon \subset U_{\epsilon}$ satisfying $C_2(m)=C_2(m_0)$. 

%In order to prove (\ref{basic}), 
To see this, set $m(\cdot +s(m))=(1+a)m_0 +y$, for some $a\in \mathbb{R}$, 
where $y\in H^1_{\alpha}$ is such that $\left<E'(m_0),\,y\right>=k\left<C_2'(m_0),\,y\right>=\int\,y\,dx=0$.
Then  Taylor's theorem with $v=m(\cdot +s(m))-m_0=a m_0+y$ gives 
$$
\begin{aligned}
%C_2(m_0)&=C_2(m)=C_2(m(x+s(m))\\
C_2(m(\cdot +s(m)) &=C_2(m_0)+\left<C_2'(m_0),\,v\right>+{O}(\|v\|^2_{\mathcal{Z}})\\
&=C_2(m_0)+\frac{a}{k}\int m_0\,dx+{O}(\|v\|^2_{\mathcal{Z}}),
\end{aligned}
$$
(where we used the fact that $C_2'(m_0)=1/k$), and also 
$C_2(m_0)=C_2(m)=C_2(m(\cdot +s(m))$
by translation invariance of $C_2$, from which it follows that 
$a={O}\left(\|v\|^2_{\mathcal{Z}}\right)$. 
A Taylor expansion of $F=k\,C_2-E$, with  $F'(m_0)=0$ and $F''(m_0)={\tt H}$, gives
$F(m)=F(m(\cdot +s(m))=F(m_0)+ \frac{1}{2} \left<{\tt H}v,\,v\right>+o\left(\|v\|^2_{\mathcal{Z}}\right)$.  
Using the fact that  $C_2(m)=C_2(m_0)$ once more, together with the estimate of the magnitude of $a$,  
the previous relation yields $ E(m_0)-E(m) = \frac{1}{2}\left<{\tt H}v,\,v\right>+o\left(\|v\|^2_{\mathcal{Z}}\right) 
= \frac{1}{2}\left<{\tt H}y,\,y\right>+o\left(\|v\|^2_{\mathcal{Z}}\right)$. 
But $\left(y,\,m_{0,x} \right)_{\al}=\left(m(\cdot +s(m)),\,m_{0,x} \right)_{\al}=0$ 
using Lemma \ref{GSSlem} and  $(m_0,m_{0,x} )_{\al}=0$.
Therefore Lemma \ref{poslem} applies to $y$, giving %we have that
$E(m_0)-E(m)\geq \frac{\zeta }{2}\|y\|_{H^1}^2+o\left(\|v\|^2_{\mathcal{Z}}\right)$, and 
$\|y\|_{H^1}=\|v-am_0\|_{H^1}\geq\|v\|_{H^1}-|a|\| m_0\|_{H^1}\geq \|v\|_{H^1}-{O}\left(\|v\|^2_{H^1}\right)$, so for $\|v\|_{\mathcal{Z}}$ small enough we have 
$
E(m_0)-E(m)\geq \frac{\zeta}{4}\|v\|_{H^1}^2,
$
which proves (\ref{basic}).

To complete the proof, suppose that $m_0$ is unstable. 
Then there exists a sequence of initial data $m_n(\cdot ,0) \in \mathcal{Z}$ 
for $n=1,2,\ldots $ 
and $\eta>0$ such that
$$
\|m_n(\cdot,0)-m_0\|_{\mathcal{Z}}\rightarrow 0\;\;{\mbox{but}}\;\;\sup_{t>0}\inf_{\xi\in{\mathbb{R}}}\|m_n(\cdot ,t)-m_0(\cdot -\xi)\|_{H^1}\geq \eta,
$$
were $m_n(\cdot ,t)$ is the solution with initial datum $m_n(\cdot ,0)$. Let $t_n$ be the first time so that
\begin{equation}
\label{eta}
 \inf_{\xi\in \mathbb{R}}
\|m_n(\cdot ,t_n)-m_0(\cdot -\xi)\|_{H^1}=\eta .
 \end{equation}
Then as $n\rightarrow \infty$, 
$E(m_n(\cdot ,t_n))=E(m_n(\cdot ,0))\rightarrow E(m_0)$, 
and 
$C_2(m_n(\cdot ,t_n))=C_2(m_n(\cdot ,0))\rightarrow C_2(m_0)$.
Picking a sequence $v_n\in {\mathcal{Z}}$ such that 
$C_2(v_n)=C_2(m_0)$ and 
$\|v_n-m_n(\cdot ,t_n)\|_{\mathcal{Z}}\rightarrow 0$, 
it follows that $\|v_n-m_n(\cdot ,t_n)\|_{H^1}\rightarrow 0$.  
Then for $\eta$ sufficiently small,  we deduce from  (\ref{basic}) 
that 
$$
 \frac{\zeta}{4}\|v_n(\cdot +s(v_n))-m_0\|^2_{H^1} 
\leq E(m_0)-E(v_n) \rightarrow 0, 
$$  
by continuity of $E$. 
By the translation invariance of the $H^1$ norm, 
this means that 
$\|v_n-m_0(\cdot -s(v_n))\|_{H^1}\rightarrow 0$, 
which further implies   
$$\|m_n(\cdot ,t_n)-m_0(\cdot -s(v_n) )\|_{H^1}\rightarrow 0.
$$ 
This contradicts  (\ref{eta}) and completes the proof.
% \end{proof}

\section{Conclusions}

We have established the stability of the lefton solution when $b<-1$. 
These results are a first step towards understanding 
how the soliton resolution conjecture, as described in \cite{tao}, should 
hold for (\ref{bfamily}); this would be consistent with the numerical results of Holm and Staley (see Figure 1). 
However, our notion of stability is rather limited, 
in that it requires solutions that are initially close 
to the lefton with respect to the Banach space norm 
$\| \cdot \|_{\cal Z}$, in order to be close in $H^1$ at subsequent 
times. We expect that 
stability should hold more generally, 
for all initial data that is close 
to the lefton in $H^1$, at least up to the blow up time \cite{zhou};  
but in that context, the methods of \cite{Grillakis87} cannot 
be applied, because $C_2$ is not defined everywhere.   
It would be interesting to carry out further numerical studies 
to test these ideas (by considering perturbations proportional 
to $v_c=e^{-c|x|}$, for instance); the numerical 
integration of (\ref{bfamily}) is a challenging problem in 
itself \cite{chertock}. 

It would also be interesting to see whether similar methods could be used to derive stability results for other ranges of $b$ values, 
and for other explicit solutions (see e.g. \cite{vak}). 
However, for $-1<b<1$ 
there is the problem that explicit analytic formulae for the ``ramp-cliff" profiles are unknown. In the peakon regime $b>1$, there is an explicit formula: the peakon solution 
is given by $u=c\exp(|x-ct|)$, with $m$ being given by a delta function. For the integrable  cases $b=2,3$ the orbital stability of the peakons has been 
proved, but the arguments used in \cite{cstrauss, lin} make essential use of some of the higher conserved quantities for the Camassa-Holm and 
Degasperis-Procesi equations, respectively. As far as we know, for other values of $b$ the only conserved quantities are $E$, $C_1$ and $C_2$, and only $E$ 
makes sense for the peakons.

\vspace{.1in}
 
\noindent {\bf Acknowledgments.} The authors acknowledge support from the 
Isaac Newton Institute for Mathematical Sciences, where discussions on the project reported in this article began.  
We are very grateful to the Mathematisches Forschungsinstitut Oberwolfach, which supported our  
Research in Pairs visit in September 2010. We would also like to thank Adrian Constantin and Walter Strauss for helpful discussions 
and correspondence on related matters, and we thank Darryl Holm and Martin Staley for permission to use Figure 1.   
S.L. gratefully acknowledges the support of the National Science Foundation through grant DMS-0908074.

\end{document}